\documentclass[%
 reprint,
superscriptaddress,
 amsmath,amssymb,
 aps,
pra,
]{revtex4-2}

\usepackage{graphicx}
\usepackage{subcaption}
\usepackage{dcolumn}
\usepackage{bm}
\usepackage[colorlinks=true,linkcolor=blue,citecolor=blue]{hyperref}

\usepackage{amsthm}
\usepackage{booktabs}
\usepackage{amsmath}
\usepackage{braket}
\usepackage{dsfont}
\usepackage[dvipsnames,table,xcdraw]{xcolor}
\usepackage{tikz-cd}
\usetikzlibrary{positioning}

\newtheorem{theorem}{Theorem}

\newcommand{\Tr}{\text{Tr}}

\newcommand{\chn}{\mathcal{E}}

\newcommand{\chnd}{\mathcal{F}} 

\newcommand{\pt}{p} 
\newcommand{\pest}[1][]{\hat{p}_{#1}} 
\newcommand{\rhoest}[1][]{\hat{\rho}_{#1}} 

\newcommand{\ret}[2]{\mathcal{R}_{#1,#2}}
\newcommand{\retd}{\ret{\chn}{\rf}} 
\newcommand{\rete}[2]{\mathcal{R}_{#1,#2}^\textsf{ext}}

\newcommand{\rf}{\gamma}
\newcommand{\rfe}{\Gamma}
\newcommand{\rfc}{\rho}

\newcommand{\psisd}{\ketbra{\psi_x}{\psi_x}}

\newcommand{\ketbra}[2]{|#1\rangle\langle#2|}

\DeclareMathOperator*{\argmin}{\arg\!\min}

\hypersetup{
    colorlinks=true,
    linkcolor=blue,
    filecolor=magenta,      
    urlcolor=blue,
    pdftitle={Overleaf Example},
    pdfpagemode=FullScreen,
    }

\begin{document}

\title{Tomographic Limits of the Petz Recovery Map}

\author{Peter Sidajaya}
\affiliation{Centre for Quantum Technologies, National University of Singapore, 3 Science Drive 2, 117543, Singapore}

\author{Clive Cenxin Aw}
\affiliation{Centre for Quantum Technologies, National University of Singapore, 3 Science Drive 2, 117543, Singapore}

\author{Mingxuan Liu}
\email{liu.mingxuan@u.nus.edu}
\affiliation{Centre for Quantum Technologies, National University of Singapore, 3 Science Drive 2, 117543, Singapore}

\author{Valerio Scarani}
\affiliation{Centre for Quantum Technologies, National University of Singapore, 3 Science Drive 2, 117543, Singapore}
\affiliation{Department of Physics, National University of Singapore, 2 Science Drive 3, Singapore 117542}

\date{\today}

\begin{abstract}
The Petz recovery map is considered one of the key candidates for the quantum analogue of Bayesian inference and Jeffrey's conditionalization. Since, there seems to be a natural connection between Bayesian inference and the notion of state tomography, it is natural to ask if the Petz recovery can be used for this latter task. In this paper, we discuss such recent results on iterated Petz recovery and relate them to Bayesian approaches to quantum state tomography. We highlight the limitations of direct Petz iteration and show how an extended Petz construction, by lifting the inference problem to a classical distribution over candidate quantum states, recovers the structure of Bayesian and maximum likelihood tomography. This provides perspective on the modifications or nuances required for a Petz approach to quantum retrodiction to perform quantum state tomography.
\end{abstract}

\maketitle

\tableofcontents

\section{Introduction}\label{sec:intro}

The modelling of states from measurement statistics is central to quantum science. This is the task of quantum state tomography (QST) \cite{QSTparis2004quantum,QSTenglert2025lecture,QSTcramer2010efficient}. Traditionally, QST involves taking the likelihood function, which gives the probability of observing the evidence given a certain estimator, and taking the estimator that maximises the function; this is the Maximum Likelihood Estimator (MLE) \cite{hradil1997quantum, QSTparis2004quantum,QSTenglert2025lecture,QSTcramer2010efficient}. Alternatively, one could also define Bayesian methods for tomography by specifying a probability function over the space of states and updating that probability with the observed evidence \cite{SchackBrunCaves2001,blume2010optimal}.

Bearing some resemblance to the task of tomography, we have \textit{retrodiction} \cite{BS21,AwBS, watanabe65,murk2026connecting,aw2026thesis}, where one makes an inference about an input given some observed evidence on the output of a process. In the classical setting, Bayesian retrodiction is given explicitly in Jeffrey's conditionalization, also known as Jeffrey's update \cite{zhou2014belief-BK,BJPsym21}. A popular choice for a quantum analogue to this is the Petz recovery map \cite{petzisking2022axioms,QPRPetzPaper,bai2025quantum,wilde-recov}. Crucially, in the Petz map, the prior is now a quantum state encoded as a density matrix, in contrast to Bayesian tomography methods, where it is a classical probability distribution over a state space. 

When conceptualised as such, retrodiction quite naturally invokes a connection to tomography. There is an intuition that iterated, rational inference on repeated evidence should lead us eventually to all that can be known about input that resulted in that evidence. Furthermore, in the informationally identifiable case, we should eventually arrive at the input itself.

Classically, this expectation is right---iterated Jeffrey's conditionalization on a non-singular channel, given a fixed repeated evidence does perform an inversion, regardless of the initialising prior so long as it has full support \cite{wasserman2004all-of-statistics,murk2026connecting}. However, in the quantum regime, Murk \textit{et al.} in \cite{murk2026connecting} showed this expectation does not hold generally when iterating the Petz recovery map. While it is the case that such an approach never decreases the likelihood of the estimator, there is no guarantee that the MLE is actually reached, unlike in the case of classical inference. Specifically, it is only guaranteed to achieve convergence to the MLE estimator whenever all rank-deficient states have lower likelihoods than the starting reference prior. We show later that such failure modes are common, especially in higher dimensions and when the target state is relatively pure.


Murk \textit{et al.}'s results in \cite{murk2026connecting} highlight where classical and quantum retrodiction depart with respect to state tomography. Here, we investigate the extent of this departure. In particular, we explain in some detail the protocols for which the Petz recovery map cannot be used for QST(Section \ref{sec:cannot}). After which, we present two protocols for which the Petz, under modification, \textit{can} be used for QST (Section \ref{sec:can}), noting their qualitative differences. After which, we make some concluding conceptual remarks (Section \ref{sec:concl}) about these results, especially with regards to making inferences in quantum theory. A summary of the results of this paper can be seen in Fig.~\ref{fig:schematic} and Table~\ref{tab:comparison}.

\begin{figure*}
    \centering
    \includegraphics[width=0.8\textwidth]{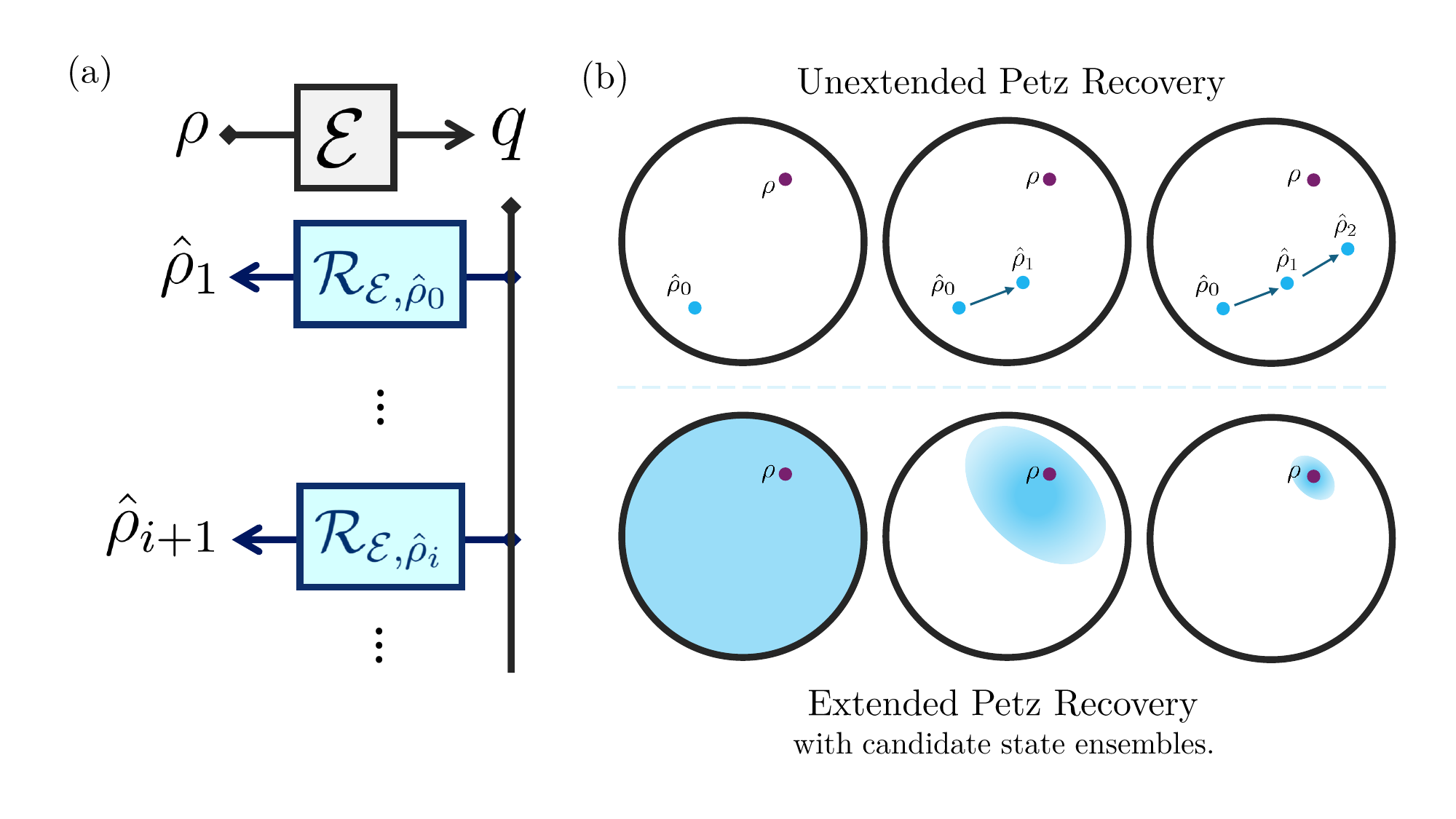}
    \caption{\textbf{(a)} \textit{The iterative (unextended) Petz protocol.} The unknown true state $\rho$ generates a classical evidence $q = \chn[\rho]$ through the measurement channel. The evidence is then retrodicted using some reference prior $\rhoest[0]$ to obtain the new updated estimator $\rhoest[1]$, which is used to perform another retrodiction in the following round. Tomography is considered successful when $\rhoest[i]$ converges to the true state $\rho$. \textbf{(b)} \textit{The reference prior in the unextended (Sec.~\ref{sec:cannot}) and the extended Petz map (Sec.~\ref{sec:can}).} In the unextended Petz map, the prior is a state in the Hilbert space, and at every iteration it moves. However, it may not always converge to the true state $\rho$. On the other hand, in the extended Petz map, the prior can be represented as a classical probability distribution (represented by the teal colouring) over the whole state space, which then shrinks towards the true state. This protocol is equivalent to standard Bayesian tomography, which reliably converges to the true state $\rho$.}
    \label{fig:schematic}
\end{figure*}

\begin{table*}[t]
\centering
\renewcommand{\arraystretch}{2}

\resizebox{0.95\textwidth}{!}{%
\begin{tabular}{lll}
\toprule
\textbf{Approach} & \textbf{Update} & \textbf{Works?} \\
\midrule
Unextended Classical
& Event-wise: $\displaystyle p_{i+1}(x)=p_i(x)\frac{\chn(y_i|x)}{\sum_{x'}p_i(x')\chn(y_i|x')}$
& No (Sec~\ref{ssec:event-unext})\\
& Batch: $\displaystyle p_{i+1}(x)=p_i(x)\sum_y q(y)\frac{\chn(y|x)}{\sum_{x'}p_i(x')\chn(y|x')}$
& Yes (Theorem~\ref{thm:c-s-conv}, Sec~\ref{ssec:safranek-qst})\\
\midrule
Unextended Petz
& Event-wise: $\displaystyle \rho_{i+1}=\frac{\sqrt{\rho_i}\,M_{y_i}\sqrt{\rho_i}}{\Tr[M_{y_i}\rho_i]}$
& No (Sec~\ref{ssec:event-unext}) \\
& Batch: $\displaystyle \rho_{i+1}=\sum_j q(j)\frac{\sqrt{\rho_i}M_j\sqrt{\rho_i}}{\Tr[M_j\rho_i]}$
& No (Murk et al. \cite{murk2026connecting}, Sec~\ref{ssec:safranek-qst})\\
\midrule
Extended Petz (CSE)
& Event-wise: $\displaystyle p_{i+1}(x)=p_i(x)\frac{\Tr[M_{y_i}\rho_x]}{\Tr[M_{y_i}\rhoest[i]]}$
& Yes, equivalent to Bayesian tomography (Sec~\ref{ssec:event-ext})\\
& Batch: $\displaystyle p_{i+1}(x)=p_i(x)\sum_j q(j)\frac{\Tr[M_j\rho_x]}{\Tr[M_j\rhoest[i]]}$
& Yes (Sec~\ref{ssec:batch-ext})
\\
\bottomrule
\end{tabular}%
}

\caption{A summary of our paper: the comparison of event-wise and batch updates in unextended classical inference, unextended Petz recovery, and extended Petz recovery with a candidate state ensemble (CSE). Here $q(j)=\Tr[M_j\rho]$. For the event-wise approach, the evidence come as individual samples from $q(j)$, labelled as $y_i$ here. For the batch approach, the evidence is $q(j)$ itself.}
\label{tab:comparison}

\end{table*}

\section{The Unextended Petz Map \\ Does Not Perform State Tomography} \label{sec:cannot}
\subsection{Preliminaries}
We begin by defining the formalism and tools we will need in this section. Since comparisons to Bayes' rule (and Jeffrey's update) will become crucial when understanding the Petz recovery, we define them here in the category of finite stochastic maps. For the classical regime, forward channels $\tilde\chn$ are given by stochastic matrices with conditional probabilities $\tilde\chn(y|x)$, with the input $x$ and output $y$ indices living in their respective state spaces, $x \in X$ and $y\in Y$.  The Bayesian inversion of such a channel is defined by applying Bayes' rule,
\begin{equation}\label{eq:classical-retrodiction}
    \ret {\tilde\chn} {\rf} (x|y):=\frac{\tilde\chn(y|x)\rf(x)}{[\tilde\chn\rf](y)}.
\end{equation}
Here, $\rf(x)$ denotes the reference prior necessarily invoked in any Bayesian inference. Additionally, to avoid division by zero we assume that $\tilde\chn$ has positive entries. $\tilde\chn\rf$ refers to the propagation of a distribution $\rf(x)$ through $\tilde\chn$. That is, $[\tilde\chn\rf](y)=\sum_x \tilde\chn(y|x)\rf(x)$. Now if one has a prior $\pest[i-1](x)$ pertaining to the input state, and observes an evidence distribution $q_{i+1}(y)$ that is on the output of $\chn$, one may use Jeffrey's conditionalization to update $\pest[i-1](x)$ to some new state $\pest[i](x)$. This is given by propagating some output evidence $q(y)$, through the Bayesian inverted map $\ret {\tilde\chn} {\rf}$: 

\begin{align}\label{eq:c-s-iterative}
    \pest[i](x) &= \sum_y \ret {\tilde\chn}{\pest[i-1]}(x|y)q_{i}(y) \\ 
    &= \hat{p}_{i-1}(x)\sum_y \frac{\tilde\chn(y|x)}{[\tilde\chn\hat{p}_{i-1}](y)} q_{i}(y)\label{eq:jeff-upd}
\end{align}

With this, consider the quantum regime. Here, the Petz recovery map takes a role analogous to Bayes' rule \cite{petzisking2022axioms,QPRPetzPaper,AwBS,liu2026-extendedpetz,bai2025quantum}. Given a forward process defined by a completely positive trace-preserving (CPTP) map, which we denote here as $\chnd:A \to A'$, and a reference state $\rf \in A$ (now, a density operator), the Petz map \cite{petz,petz1,wilde_2013} is given by
\begin{equation}\label{eq:petz-map}
    \ret\chnd\rf [\omega]=\sqrt{\rf}\chnd^\dag\left[\frac{1}{\sqrt{\chnd[\rf]}}\omega\frac{1}{\sqrt{\chnd[\rf]}}\right]\sqrt{\rf}.
\end{equation}

Here, we have a kind of quantum Jeffrey's conditionalization with $\omega \in A'$ as a quantum theoretic evidence on the output space used to update $\rf$ to $\ret\chnd\rf[\omega]$. In order to emphasize the similarities across the regimes, we have chosen to have some overlap in notation, with some labels used for classical and quantum objects playing the same conceptual role (i.e. $\mathcal{R}$ for the object corresponding to retrodiction, $\rf$ for one's reference prior). They will be disambiguated by context.  

Since we focus on the task of tomography, we employ the following conventions with respect to \textit{measurement} channels. Let the measurements be described with a set of POVM elements ${M_j}$, where $\sum_j M_j = \mathbb{I}$. Then, we may write the measurement channel simply as
\begin{equation}\label{eq:mchn-povm}
    \chn[\rho] = \sum_j \Tr\big[M_j\rho\big] \ketbra{j}{j}.
\end{equation}
This is a quantum-to-classical channel, since the output is composed by a diagonal classical register. With this defined, we have the adjoint as follows,
\begin{equation}
    \chn^\dagger [\cdot] = \sum_j M_j \bra{j} \cdot \ket{j}, 
\end{equation}
and thus, we also have the corresponding Petz recovery on a classical output, expressed here as an operator $q=\sum_yq(y) \ketbra{y}{y}$ \cite{fuchs2003quantum,liu2025retrodictive}:
\begin{equation} \label{eq:ret-mchn-povm}
    \retd[q]=\sum_{j} q(j) \frac{\sqrt{\gamma}M_j\sqrt{\gamma}}{\Tr(M_j\gamma)}.
\end{equation}
Of particular interest in tomography are informationally complete POVMs (IC-POVM), which are POVMs that have at least $d^2$ elements and span the whole space of Hermitian operators. In the usual well-behaved scenarios, the use of IC-POVM in tomography leads to a unique MLE estimator.

With this, one can take two approaches to attempting tomography via the Petz recovery map. We go through both now, showing that neither of them give us a generally sound route to quantum state tomography.


\subsection{Event Approach}\label{ssec:event-unext}
The first is a na\"{i}ve application of Jeffrey's update that is straightforward to exclude for tomography. This failure will be of note for later discussions, so we describe it in some detail here. In this approach, one sends copies of a quantum state $\rho$ from some source through a measurement channel $\chn$. Singular click events of the classical state indexed in $j$ are registered. Accordingly, for $i \in \mathbb{Z}^{>0}$ this is used to update the prior $\rhoest[i-1]$ to $\rhoest[i]$ at every $i$-th event. With the Petz recovery, then, if we get $[q_i]=k$ (i.e. $q_i = \ketbra kk$) at the $i$-th click on the classical register, we have:
\begin{equation}
    \rhoest[i] \overset{[q_i]=k}{=}  \ret\chn{\rhoest[i-1]}\big[\ketbra{k}{k}\big], \quad k  \sim P(k|\rho,\chn)=\Tr\big[M_k\rho\big]
\end{equation}
$M_k$ being the $k$-th POVM for $\chn$ as per \eqref{eq:mchn-povm}. The hope is that as $i$ gets larger, one converges to some kind of MLE for $\rho$, and that for informationally complete $\chn$ we get $\rhoest[N] \to \rho$ for large $N$. 

It is easy to show that this method will simply fail, not only in quantum, but also classically. Consider the equivalent classical problem, which we include here briefly, where we have a state $p$ that enters a channel $\chn$, and what we receive are samples from $q=\chn p$. Now, the updates are given by
\begin{equation}
    \pest[N](x) = \frac{\pest[0](x)\prod_i^N \chn(q_i|x)}{\sum_{x'}\pest[0](x')\prod_i^N \chn(q_i|x')}.
\end{equation}
Now, note that for injective channels, a distribution that is invariant under every possible event-wise update must be deterministic, $\pest[N](x)=\delta_{x,k}$. Thus, supposing that the update does converge, it will only converge to these deterministic distributions, and these are generally very poor estimators for the target state $p$. Bringing this back to quantum, only pure states are fixed points. This can be seen from Eq.~\eqref{eq:ret-mchn-povm}, where forcing $\rhoest[i+1]=\rhoest[i]$ requires $\rhoest[i+1]\propto\sqrt{\rhoest[i]}M_j\sqrt{\rhoest[i]}$ for all $j$. When the POVMs are IC, this can only be possible when $\rhoest[i]$ is pure.

In particular, we have treated our belief as a particular \textit{state} $p\in\Delta$. However, a more appropriate use of Bayesian inference instead encodes our belief as a distribution over the simplex $\Delta$ (see Appendix \ref{app:coin} for the biased coin example), which foreshadows our Section~\ref{sec:can}.

\subsection{Batch Approach}\label{ssec:safranek-qst}
Having settled that event-wise approach to the Petz map does not work, we may consider the other option. Here, we send the evidence as measurement frequencies from the population, as opposed to singular incoming samples \footnote{For completeness, one can consider a streaming version of this scenario. We may think of $\mathbf{N}_i = (N_{i,0},N_{i,1},\ldots,N_{i,m})$ as the accumulated counts of the observed outcomes up to the $i$-th event. Then, ${\hat q_i = \mathbf{N}_i/( \Sigma_{j}  N_{i,j})}$
denotes the corresponding empirical distribution. Upon observing a new click $k$, counts are updated, $\mathbf{N}_{i+1} = \mathbf{N}_i + \mathbf{e}_k$,
where $\mathbf{e}_k$ is a pure vector for $k$. Thus, writing $N_i = \Sigma_j N_{i,j}$, the empirical distribution evolves as
\begin{equation}
    \hat{q}_{i+1}
    = \frac{N_i}{N_i+1}\hat{q}_i
    + \frac{1}{N_i+1}\delta_k.
\end{equation}
$\hat q_{i+1}$ may then be converted into a density operator $q$ that can be fed into the Petz recovery given by \eqref{eq:ret-mchn-povm}. However, by the law of large numbers \cite{wasserman2004all-of-statistics}, the long run asymptotic behaviour of retrodicting on such streaming, accumulative evidence is the same as retrodicting on a fixed evidence characterized by $q=\chn[\rho]$.}. Here, to simplify the analysis, we will only consider the behaviour of the protocol in an asymptotic regime (i.e. $q = \chn[\rho]$), such that we would not need to worry about sampling noise. In this regime, the state with maximum likelihood, or simply the MLE state, will coincide with the `true' unknown state when the channel $\chn$ is injective.

Now, as observed by Murk \textit{et al.} \cite{murk2026connecting}, convergence to a genuine MLE may not occur under such a setting, and so this protocol is also generally unfit for tomography. In particular, when we have
\begin{equation}\label{q-s-iterative}
    \rhoest[i] = \ret{\chn}{\rhoest[i-1]}[q],
\end{equation}
defined for $i \in \mathbb Z^{>0}$ and some initializing $\rhoest[0]$. Then, \cite{murk2026connecting} showed two crucial theorems:
\begin{theorem}
    \cite{murk2026connecting} The iteration of the Petz map of the measurement channel $\chn$ with the estimator at time $i$, the prior $\rhoest[i]$ never decreases the likelihood:
    \begin{equation}
        \mathcal{L}(\rhoest[i]) \geq \mathcal{L}(\rhoest[i-1]),
    \end{equation}
    with equality if and only if $\rhoest[i]$ is a fixed point.
\end{theorem}

\begin{theorem}
    \cite{murk2026connecting} Suppose that $\{M_k\}$ is a tomographically complete measurement with all observed probabilities $\{q_i(k)\}$ positive. We denote the unique maximizer of the log-likelihood $\rhoest[MLE]$. For some initial prior $\rhoest[0]$. Then $\rhoest[i] \rightarrow \rhoest[MLE]$ or $\det(\rhoest[i]) \rightarrow 0$. If the upper level set $\{\rhoest : \mathcal{L}(\rhoest) \geq \mathcal{L}(\rhoest[0])\}$ contains only invertible states, then it will converge to $\rhoest[MLE]$.
\end{theorem}

That is, the iterative Petz map for measurement process is indeed monotonic in likelihood with respect to the data, but it may fail by converging to a rank-deficient state instead of the MLE estimator. In contrast to the event approach, the failure in this batch evidence setting is more surprising, because the classical version of the batch approach \textit{does} converge to the MLE estimator (and thus to the true state $p$ if the map is injective).

\begin{theorem}\label{thm:c-s-conv}
    In the limit of many iterations, and given a full-support initial prior $\pest[0]$, Jeffrey's conditionalization \eqref{eq:jeff-upd} makes the estimator $\pest[i]$ converge to the MLE estimator. In particular, for a given injective map $\chn$ with all positive entries, actual input $\pt$, and $q(y)=[\chn p](y)$:
    \begin{equation}
        \lim_{i\rightarrow\infty}\pest[i](x)=\pt(x).
    \end{equation}
\end{theorem}
\begin{proof}
    The proof can be found in Appendix~\ref{app:c-s-conv}.
\end{proof}

This is a notable gap between classical Bayesian updating and the Petz recovery map. 

\subsubsection*{The Mirrored Petz and $R\rho R$ Tomography}

Additionally, we note the following observation: let
\begin{equation}
    R_{\rhoest}=\chn^\dag\left[\frac{1}{\sqrt{\chn[\rhoest]}}q\frac{1}{\sqrt{\chn[\rhoest]}}\right].
\end{equation}
Then, the iterative Petz update can be written as
\begin{equation}
    \rhoest[i+1] = \sqrt{\rhoest[i]} R_{\rhoest[i]} \sqrt{\rhoest[i]}.
\end{equation}
Now, if we factorise it to $(\sqrt{\rhoest[i]R_{\rhoest[i]}})(\sqrt{R_{\rhoest[i]}\rhoest[i]})$ and instead reverse the ordering to the two factors, we get
\begin{equation}
    \rhoest[i+1] = \sqrt{R_{\rhoest[i]}} \rhoest[i] \sqrt{R_{\rhoest[i]}}.
\end{equation}
Let us call this update the \textit{mirrored Petz map}. Expectedly, this map gives different updates to the Petz recovery. Surprisingly, for all our numerical runs, the iterative mirrored Petz update does converge to the true state, including cases where the ordinary Petz map fails to. Furthermore, this update is evocative of another update:
\begin{equation}
    \rhoest[i+1] = \mathcal{N}\left[{R_{\rhoest[i]}} \rhoest[i] {R_{\rhoest[i]}}\right],
\end{equation}
where $\mathcal{N}$ is normalisation. This update is the $R\rho R$ algorithm, a well known algorithm to find the MLE estimator going back to \cite{hradil1997quantum,jevzek2003quantum}. Unlike the iterative Petz map, this update is not monotonic in the likelihood \cite{vrehavcek2007diluted}, and it does have cases where it fails to converge to the MLE estimator \cite{oberender2025spurious}. However, numerically it seems much more reliable than the iterative Petz map, as we will show in Section~\ref{sec:numerics}. A full treatment of the behaviour of the mirrored Petz and its relation to the $R\rho R$ algorithm is beyond the scope of this paper, but Appendix~\ref{app:failure} provides a sketch explaining why the mirrored Petz map converges even where the iterative Petz fails.

Whether any update rule could possibly be used for tomography via such a protocol (where the Bayesian prior is only encoded as a density operator in the input state space) is yet to be seen. In the next section, we show instead an adjustment on the Petz recovery that does allow for tomographic uses.


\section{The Extended Petz Map \\ for State Tomography}\label{sec:can}

We move to introduce tools that will enable the Petz map for tomography.

\subsection{Preliminaries}

It was noticed in Ref.~\cite{liu2026-extendedpetz} that quantum retrodiction generally requires the prior to specify one's belief not only on the system of interest, but also on the preparation. With this in mind, we first introduce the notion of an \textit{extended Petz recovery map}
\begin{equation}
    \rete \chnd \rfe [\omega] = \sqrt{\rfe} \left( \chnd^\dagger\left[\frac{1}{\sqrt{\chnd[\rf]}} \omega \frac{1}{\sqrt{\chnd[\rf]}} \right] \otimes \mathbb{I}_{B} \right)\sqrt{\rfe}.
\end{equation}
Whereas in the typical Petz map we have $\chnd : A \rightarrow A'$ and a reverse map $\ret{\chnd}{\rf} : A' \rightarrow A$, here we have an extended space $B$. The extended Petz map is then obtained by applying the Petz map to the combined channel
\begin{equation}
    \Tr_B \circ(\chnd_A \otimes \mathbb{I}_B).
\end{equation}
That is, the forward channel is a channel acting trivially on the extended space $B$, which is eventually ignored. In general, we will have the following state space conventions: $\rfe \in AB,$ is the extended reference prior, and taking the partial trace $\Tr_B[\rfe] = \rf$ gives us a prior in the original state space. We also introduce the notion of a \textit{candidate state ensemble} (CSE):
\begin{equation}
    \{p(x), \rho_x\}_{x\in X},
\end{equation}
where $X$ is the parameter space, $p(x)$ is a probability distribution over that space, and $\rho_x$ is the associated state to that parameter. For example, in an unconstrained qubit tomography $X$ is the Bloch ball and $\rho_x$ is simply the density operator at the point $x\in X$ in the Bloch ball. Notice how a CSE determines a mean state
\begin{equation}
    \rhoest = \int dx \; p(x) \rho_x,
\end{equation}
but not vice versa. CSEs are thus a classical belief over how a given quantum system ought to be described, and $\rhoest$ is the estimator that is obtained from the belief. This is identical to the description of reference priors in Bayesian tomography \cite{SchackBrunCaves2001,CavesFuchsSchack2002dF}.

Now, to connect the classical reference prior with the extended Petz map, consider the following extended reference prior $\rfe_i$ (which is supplied with an index $i$, to anticipate update iterations later on): 
\begin{equation}
    \rfe_i = \int dx \; p_i(x) \ket{\psi_x}\bra{\psi_x}_{AR} \otimes \ket{x} \bra{x}_C,
\end{equation}
with the extension space $B \equiv RC$ now containing two component spaces: one being a purifying reference space $R$ such that $\psisd \in AR$, with $\Tr_R \big[\psisd\big]= \rfc_x$, and the other one being a classical state index $C$ for the classical register $\ketbra x x$. Note that the mean state is simply obtained by tracing out $B$: $\Tr_B[\rfe_i]=\rhoest[i]$. This can then be used to acquire an extended Petz recovery on a measurement channel $\chn$ \cite{liu2025unifying}:
\begin{widetext}
    \begin{equation}
    \rete \chn {\rfe_i} [q_i] = \sqrt{\rfe_i}  
    \left( \sum_j M_j \bra{j}_{A'} 
    \frac{1}{\sqrt{\chn[\rhoest[i]]}} q_i
    \frac{1}{\sqrt{\chn[\rhoest[i]]}} \ket{j} _{A'}
    \otimes \mathbb{I}_{RC} \right)\sqrt{\rfe_i},
\end{equation}
with $q_i= \sum_y q_i(y) \ketbra yy_{A'}$, this reduces to
\begin{equation} \label{eq:ret-chn-ext-tomo}
    \rete \chn {\rfe_i} [q]  = \int dx\sum_{j} q(j) p_i(x) \frac{\Tr[M_j \rfc_x]}{\Tr[M_j \rhoest[i]]} \psisd \otimes \ketbra x x .
\end{equation}
\end{widetext}
With this, we can apply these tools for tomography.

\subsection{Event Approach}\label{ssec:event-ext}


For the event approach, we receive singular clicks of a classical register as we send $\rho$ through $\chn$. If we register $[q_i]=k$ (i.e. $q_i = \ketbra kk$) at the $i$-th click (for $i \in \mathbb{Z}^{>0}$), we update $\rfe_{i-1}$ as follows
\begin{equation}
    \rfe_{i}\overset{[q_i]=k}{=} \int dx \; p_{i-1}(x) \frac{\Tr[M_k \rfc_{x}]}{\Tr[M_k \rhoest[i-1]]} \psisd \otimes \ketbra x x .
\end{equation}
This simply means that every collection of data adds a factor corresponding to the ratio of the probabilities that the component candidate triggers $k$ and the probability that the overall mean state triggers $k$:
\begin{equation}
    i \in \mathbb{Z}^{>0}: \quad p_{i}(x) \overset{[q_i]=k}{=}  p_{i-1}(x)  \frac{\Tr[M_k \rfc_{x}]}{\Tr[M_k \rhoest[i-1]]}
\end{equation}
Together, we can then write the eigenvalues of $\rfe_i$ as a function of a set of collected data $\Vec{q_i} = ([q_1],[q_2]\dots [q_i])$.
\begin{equation}\label{eq:event-update-extended}
    p_i(x|\Vec{q_i})= \frac{p_0(x)\prod^i_{\ell=1}\Tr[M_{[q_\ell]} \rfc_x]}{\int dx' p_0(x')\prod^i_{\ell'=1}\Tr[M_{[q_{\ell'}]} \rfc_{x'}]},
\end{equation}
where $p_0$ is the initializing reference distribution characterizing $\rfe_0$. Thus, 
\begin{equation}
    \rfe_i= \int dx \;  p_i(x|\Vec{q_i}) \ketbra{\psi_x}{\psi_x}\otimes \ketbra{x}{x},
\end{equation}
and taking a partial trace we have an estimator for $\rho$,
\begin{equation}
    \rhoest[i]=\Tr_{RC}[\rfe_i]= \int dx \;  p_i(x|\Vec{q_i}) \rfc_x.
\end{equation}
This expression, it turns out, is equivalent to the estimator obtained in \cite{SchackBrunCaves2001}, which is a Bayesian tomography protocol applicable to the case of quantum states generally. This is easily seen from Eq.~\eqref{eq:event-update-extended}, where the numerator is simply the likelihood of the data $\Vec{q_i}$ given a certain parameter $x$,
\begin{equation}
    \mathcal{L}(\Vec{q}_i|\rho_x)=\prod^i_{\ell=1}p(q_\ell|\rho_x)=\prod^i_{\ell=1}\Tr[M_{[q_\ell]} \rfc_x],
\end{equation}
and the update is simply
\begin{equation}
    p_i(x|\Vec{q}_i) \propto p_0(x) \mathcal{L}(\Vec{q}_i|\rho_x),
\end{equation}
which is the standard updating in Bayesian tomography.

\subsection{Batch Approach}\label{ssec:batch-ext}
While the event approach worked, we would still like to find out whether a batch approach, similar to Section~\ref{ssec:safranek-qst}, would also work. Thus, let us consider now an approach where the incoming data $q$ is the exact measurement frequencies. i.e, $q = \chn[\rho]$. Applying \eqref{eq:ret-chn-ext-tomo}, we have
\begin{align}
&\rete \chn {\rfe_{i-1}} \big[\chn[\rho]\big] \nonumber
= \\
&\int dx \sum_{j}
p_{i-1}(x)
\Tr[M_j\rho]
\frac{\Tr[M_j\rfc_x]}
{\Tr[M_j\rhoest[i-1]]}
\psisd\otimes\ketbra{x}{x}.
\end{align}
Taking the partial trace over $RC$ gives the update
\begin{equation}\label{eq}
\rhoest[i]
=
\int dx\sum_{j}
p_{i-1}(x)
\Tr[M_j\rho]
\frac{\Tr[M_j\rfc_x]}
{\Tr[M_j\rhoest[i-1]]}
\rfc_x,
\end{equation}
or equivalently,
\begin{equation}\label{eq}
p_{i}(x)=p_{i-1}(x)\sum_j\Tr[M_j\rho]\frac{\Tr[M_j\rfc_x]}{\Tr[M_j\rhoest[i-1]]}.
\end{equation}
This is nothing but a classical Jeffrey's update on the evidence $q(j)=\Tr[M_j\rho]$, as in \eqref{eq:jeff-upd}. In particular, we may write as a classical stochastic channel $\Tr[M_j\rfc_x] = \tilde{\chn}(j|x)$ with $\sum_j\tilde{\chn}(j|x) = 1$, the prior as $p_{i-1}(x)$ and the denominator terms as a propagation $\Tr[M_j\rhoest[i-1]] = \sum_x p_{i-1}(x) \Tr[M_j \rho_x ] = [\tilde{\chn}p_{i-1}](j)$.

By Theorem~\ref{thm:c-s-conv}, this series of iterations will indeed bring the likelihood of $\rhoest[i]$ to the supremum \footnote{Given that the true state is inside the convex hull of the CSE. In this paper, we consider a continuous CSE that covers the whole space which automatically guarantees it. However, this detail is relevant in the implementation, where there is only a finite number of $\rho_x$.}. When the POVM is informationally complete, the MLE is unique. However, because the map from $\pest$ to $\rhoest$ is not injective, i.e., different distributions of belief can lead to the same mean ensemble state, the underlying classical belief itself is not unique. That is, there can be one unique $\rhoest[MLE]$ that can be achieved by different $\pest$, because of the fact that likelihood is only dependent on $\rhoest$, not $\pest$. In any case, the  estimator indeed converges to the target state as the iterations increase, as long as the measurements are informationally complete. 



\section{Numerical Examples}\label{sec:numerics}

\begin{figure}
    \centering
    \includegraphics[width=0.45\textwidth]{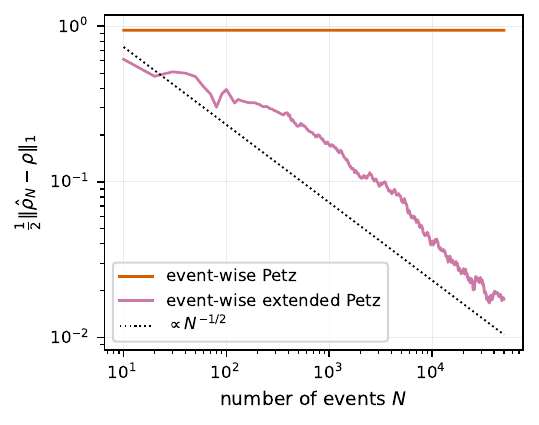}
    \caption{The comparison of the event-wise methods for a random example in $d=4$. Even with non-projective measurements, the unextended Petz map converged to a pure state relatively quickly, as can be seen by the plateau in the distance. The extended Petz map, which is equivalent to Bayesian tomography, converges to the true state as expected.}
    \label{fig:comparison_event}
\end{figure}

\begin{figure}
    \centering
    \includegraphics[width=0.45\textwidth]{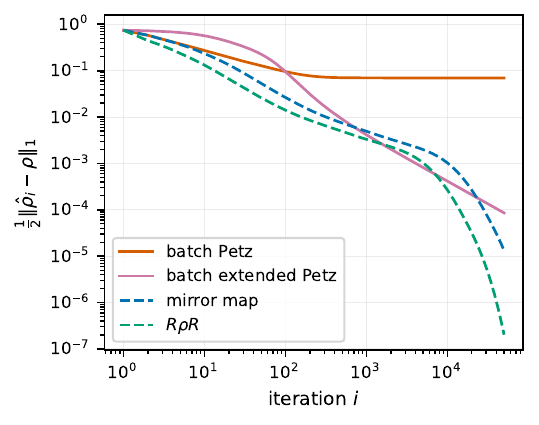}
    \caption{The comparison of the batch methods for an instance in $d=4$ where the iterative unextended Petz map gets stuck in a rank-deficient state, as measured by the distance to the true state. The other three methods, on the other hand, converge to the true state as expected.}
    \label{fig:comparison_batch}
\end{figure}

\begin{figure*}
    \centering
    \includegraphics[width=0.95\textwidth]{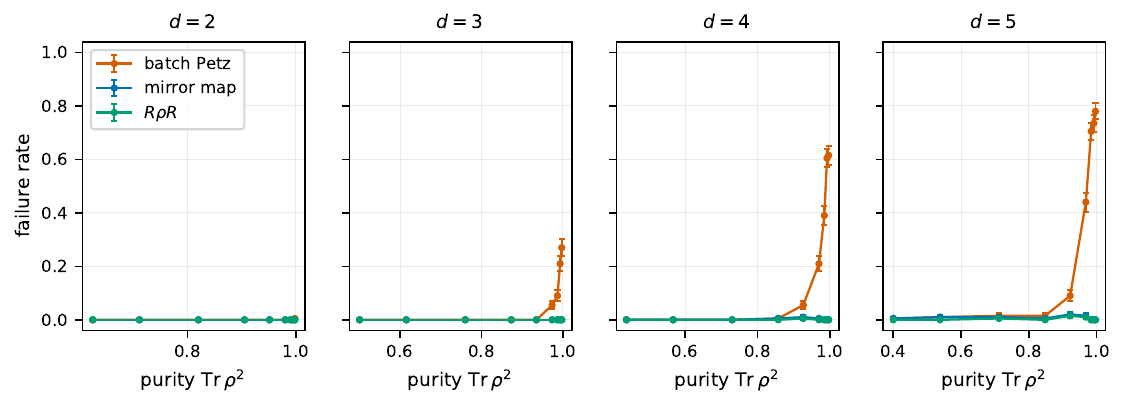}
    \caption{The failure rate of the iterative unextended Petz map, along with the mirrored Petz map and the $R\rho R$ for different dimensions. The condition of failure is $\frac{1}{2}\|\rhoest[N]-\rho\|_1>0.01$ after 50000 iterations. As can be seen, the Petz map tends to fail more often at higher dimensions and high purity. Note that due to the basic nature of the failure condition, false negatives and false positives should be expected.}
    \label{fig:capture_rate}
\end{figure*}

To illustrate the performance of the different protocols, we run each of their respective algorithms for the same initial setup (i.e. the same POVMs and input states and so on). Once again, the protocols are:
\begin{enumerate}
    \item \textit{The event-wise unextended Petz map} (Sec.~\ref{ssec:event-unext}): this converges to a pure state quickly, even for non-projective measurements, since the only fixed points are the pure states.
    \item \textit{The batch unextended Petz map} approach of \cite{murk2026connecting} (Sec.~\ref{ssec:safranek-qst}): this one often fails when the true state has high purity and when the dimension is high.
    \item \textit{The event-wise extended Petz map} (Sec.~\ref{ssec:event-ext}), equivalent to Bayesian tomography: by virtue of the extension, and the fact that the belief converges to a Dirac-delta, one needs to be careful in the implementation, which necessarily requires a discretization. The details of the implementation can be found in \cite{granade2016practical}.
    \item \textit{The batch extended Petz map} (Sec.~\ref{ssec:batch-ext}): in contrast to the event-wise counterpart, discretization is more straightforward in this setting; as long as the true state is in the convex hull of the discrete points, the estimator will converge to the true state.
\end{enumerate}
Additionally, we include comparisons between the mirrored Petz map and the $R\rho R$ maps to the unextended Petz map. Comparisons for one instance in both approaches can be found in Fig.~\ref{fig:comparison_event} and Fig.~\ref{fig:comparison_batch}. These show the unextended Petz map does not perform tomography, for both event-wise and batch approaches.

Note that the example in Fig.~\ref{fig:comparison_batch} is specifically selected as an example in which the unextended Petz map fails, since the failure is not guaranteed. To complete the comparison, in Fig.~\ref{fig:capture_rate} we analyse the failure rate of the unextended Petz map in the batch approach through different dimensions and different purities of the true state. As a comparison, we also added the mirrored Petz map and the $R\rho R$ algorithm. For each randomly sampled instance (consisting of $d^2+2$ random POVMs and a randomly sampled state with a certain purity), we run the three protocols for 50000 iterations. If $\|\rhoest[N]-\rho\|_1>0.01$ at the end, we flag that particular instance as having failed to converge by the protocol. This method is rudimentary and does not fully capture the failure rate (as there can be false positives and false negatives); however, it is enough to illustrate the staggering difference between the three maps and the tendency of the normal Petz map to fail at high purity.

In Appendix~\ref{app:failure}, we propose an explanation of why the Petz map fail at this regime. Moreover, although here we use purity, what is more important is how close the true state is to the rank-deficient subspace, which in $d>2$ can be achieved without having high purity. Especially in higher dimensions, counterexample could be found where the true state is close to a rank-deficient subspace, and thus the Petz map fails to converge, while having a relatively low purity. However for our illustration purity is a good proxy for this behaviour. Finally, considering that operationally purity is desired in quantum algorithms and protocols, this failure point seems to disqualify the iterative unextended Petz map as a \textit{useful} tomography protocol without any modifications.

\section{Concluding Remarks}\label{sec:concl}
Here, we discuss possible upshots of our observations.  We first note how the obstruction to using the unextended Petz recovery map for tomography seems to lie ultimately in a mismatch between the space on which the inference is performed and the object being estimated. In the classical case, the na\"ive form of event-wise Jeffrey's updating similarly fails when performed directly on the simplex of states (that is, the direct classical parallel to the protocol in Section \ref{ssec:event-unext}). This can be remedied by extending the inference to a parameter space (eg. the bias of a coin, see Appendix \ref{app:coin}), with each point in the simplex corresponding to a candidate with a live probability weight, as it were. 

The extended Petz construction follows precisely this logic: the CSE lifts the inference from the density operator $\rho$ to a classical belief $p_i(x)$ over candidate states $\rho_x$, with $\rhoest[i]$ recovered only as its mean. In this sense, the modifications are not arbitrary, but a quantum counterpart of the parameter-space approach that makes classical, event-wise Bayesian construction of probability distributions possible.

This also leaves an interesting distinction between the classical and quantum settings. On the classical simplex, the corresponding likelihood maximization by repeatedly performing Jeffrey's updates (as a na\"ive parallel to the protocol in Section \ref{ssec:safranek-qst}) does produce a MLE under the usual non-singularity assumptions. For quantum states, however, the iteration may instead approach rank-deficient states without reaching the MLE \cite{murk2026connecting}. Why this occurs remains an open question. It may suggest that one's model or belief about a quantum system cannot be solely captured by a density operator in a corresponding state space. Rather, one invokes a duality to one's model: retaining an open \textit{classical} belief over purified \textit{quantum} states, precisely as encoded by a CSE. The extended Petz map then provides the natural retrodictive mechanism on this enlarged space.

Finally, the Petz perspective clarifies the tomographic capabilities and limitations of quantum retrodiction. It shows that the Bayesian interpretation of the Petz map can indeed be applied to tomography (recovering the key protocol found in \cite{SchackBrunCaves2001} and inducing a generally applicable batch-wise method as well) but only after augmenting the state space in a manner sensitive to the geometry of quantum states. The extended Petz map therefore provides a bridge between quantum retrodiction and Bayesian tomography, while also making explicit why the direct density-operator analogue of classical Bayesian updating may be fundamentally insufficient. \\ \vspace{1em}

\section*{Acknowledgements}
We thank Francesco Buscemi, Ge Bai, Seok Hyung Lie, and Robert Spekkens for helpful discussions. This project is supported by the National Research Foundation, Singapore through the National Quantum Office, hosted in A*STAR, under its Centre for Quantum Technologies Funding Initiative (S24Q2d0009); and by the Ministry of Education, Singapore, under the Tier 2 grant ``Bayesian approach to irreversibility'' (Grant No.~MOE-T2EP50123-0002).

\bibliographystyle{unsrt}
\bibliography{apssamp}

@PREAMBLE{
 "\providecommand{\noopsort}[1]{}" 
 # "\providecommand{\singleletter}[1]{#1}%" 
}

@article{bai2025quantum,
  title={Quantum bayes’ rule and petz transpose map from the minimum change principle},
  author={Bai, Ge and Buscemi, Francesco and Scarani, Valerio},
  journal={Physical Review Letters},
  volume={135},
  number={9},
  pages={090203},
  year={2025},
  publisher={APS}
}

@article{AwBS,
author = {Aw,Clive Cenxin  and Buscemi,Francesco  and Scarani,Valerio },
title = {Fluctuation theorems with retrodiction rather than reverse processes},
journal = {AVS Quantum Science},
volume = {3},
number = {4},
pages = {045601},
year = {2021},
doi = {10.1116/5.0060893},
URL = {https://doi.org/10.1116/5.0060893},
eprint = {https://doi.org/10.1116/5.0060893}
}

@Article{BJPsym21,
AUTHOR = {Barnett, Stephen M. and Jeffers, John and Pegg, David T.},
TITLE = {Quantum Retrodiction: Foundations and Controversies},
JOURNAL = {Symmetry},
VOLUME = {13},
YEAR = {2021},
NUMBER = {4},
ARTICLE-NUMBER = {586},
URL = {https://www.mdpi.com/2073-8994/13/4/586},
ISSN = {2073-8994},
DOI = {10.3390/sym13040586}
}

@article{BS21,
  title = {Fluctuation theorems from Bayesian retrodiction},
  author = {Buscemi, Francesco and Scarani, Valerio},
  journal = {Phys. Rev. E},
  volume = {103},
  issue = {5},
  pages = {052111},
  numpages = {8},
  year = {2021},
  month = {May},
  publisher = {American Physical Society},
  doi = {10.1103/PhysRevE.103.052111},
  url = {https://link.aps.org/doi/10.1103/PhysRevE.103.052111}
}

@article{petz1,
  title = {Sufficient subalgebras and the relative entropy of states of a von Neumann algebra},
  author = {Petz, Denes},
  journal = {Comm. Math. Phys.},
  volume = {105},
  issue = {1},
  pages = {123--131},
  year = {1986},
  doi = {10.1007/BF01212345}
}

@article{petz,
	author = {Petz, Denes},
	title = "{Sufficiency of channels over von Neumann algebras}",
	journal = {The Quarterly Journal of Mathematics},
	volume = {39},
	number = {1},
	pages = {97-108},
	year = {1988},
	month = {03},
	issn = {0033-5606},
	doi = {10.1093/qmath/39.1.97}
}

@article{watanabe65,
	title = {Conditional Probabilities in Physics},
	author = {Watanabe, Satosi},
	journal = {Progr. Theor. Phys. Suppl.},
	volume = {E65},
	issue = {},
	pages = {135--160},
	numpages = {0},
	year = {1965},
	month = {Jan},
	publisher = {},
	doi = {https://doi.org/10.1143/PTPS.E65.135}
}

@book{wilde_2013, place={Cambridge}, title={Quantum Information Theory}, DOI={10.1017/CBO9781139525343}, publisher={Cambridge University Press}, author={Wilde, Mark M.}, year={2013}}

@article{wilde-recov,
	doi = {10.1098/rspa.2015.0338},
	year = 2015,
	volume = {471},
	pages = {20150338},
	author = {Wilde, M.M.},
	title = {Recoverability in quantum information theory},
	journal = {Proceedings of the Royal Society A}
}

@article{petzisking2022axioms,
  title={Axioms for retrodiction: achieving time-reversal symmetry with a prior},
  author={Parzygnat, Arthur J and Buscemi, Francesco},
  journal={Quantum},
volume = {7},
pages = {1013},
  year={2023}
}

@article{QPRPetzPaper,
  title={Quantum Bayesian Inference in Quasiprobability Representations},
  author={Aw, Cenxin Clive and Onggadinata, Kelvin and Kaszlikowski, Dagomir and Scarani, Valerio},
  journal={PRX Quantum},
  volume={4},
  number={2},
  pages={020352},
  year={2023},
  publisher={APS},
url={https://doi.org/10.1103/PRXQuantum.4.020352}}

@article{aw2026thesis,
  title={Classical and quantum reverse processes through Bayesian inference},
  author={Aw, Clive Cenxin},
  journal={International Journal of Quantum Information},
  year={2026},
  publisher={World Scientific}
}

@inproceedings{zhou2014belief-BK,
  title={Belief-Kinematics Jeffrey-s Rules in the Theory of Evidence.},
  author={Zhou, Chunlai and Wang, Mingyue and Qin, Biao},
  booktitle={Conference on Uncertainty in Artificial Intelligence},
  year={2014},
  url={https://api.semanticscholar.org/CorpusID:16595332}
}

@article{liu2026-extendedpetz,
  title={Proper and Improper Mixed States Serve as Different Prior Beliefs for Quantum State Retrodiction},
  author={Liu, Mingxuan and Scarani, Valerio and Bai, Ge},
  journal={Physical Review Letters},
  volume={136},
  number={6},
  pages={060203},
  year={2026},
  publisher={APS}
}

@article{SchackBrunCaves2001,
  title = {Quantum {B}ayes rule},
  author = {Schack, R. and Brun, T.~A. and Caves, C.~M.},
  journal = {Phys. Rev. A},
  volume = {64},
  issue = {1},
  pages = {014305},
  year = {2001},
  publisher = {American Physical Society},
  doi = {10.1103/PhysRevA.64.014305},
}

@article{CavesFuchsSchack2002dF,
  title = {Unknown quantum states: The quantum de {F}inetti representation},
  author = {Caves, C.~M. and Fuchs, C.~A. and Schack, R.},
  journal = {J. Math. Phys.},
  volume = {43},
  number = {9},
  pages = {4537--4559},
  year = {2002},
  doi = {10.1063/1.1494475},
}

@article{murk2026connecting,
  title={Connecting Quantum Tomography and Quantum Retrodiction},
  author={Murk, Sebastian and Tan, Ian and M{\"u}ller, Fabian and {\v{S}}afr{\'a}nek, Dominik},
  journal={arXiv preprint arXiv:2606.23777},
  year={2026}
}

@book{wasserman2004all-of-statistics,
  title={All of statistics: a concise course in statistical inference},
  author={Wasserman, Larry},
  volume={26},
  year={2004},
  publisher={Springer}
}

@article{csiszar1984information,
  title={Information geometry and alternating minimization procedures},
  author={Csisz{\'a}r, Imre and Tusnády, Gábor},
  journal={Statistics and Decisions, Dedewicz},
  volume={1},
  pages={205--237},
  year={1984},
  publisher={Munich, Oldenburg Verlag}
}

@article{jevzek2003quantum,
  title={Quantum inference of states and processes},
  author={Je{\v{z}}ek, Miroslav and Fiur{\'a}{\v{s}}ek, Jarom{\'\i}r and Hradil, Zden{\v{e}}k},
  journal={Physical Review A},
  volume={68},
  number={1},
  pages={012305},
  year={2003},
  publisher={APS}
}

@article{granade2016practical,
  title={Practical bayesian tomography},
  author={Granade, Christopher and Combes, Joshua and Cory, DG},
  journal={new Journal of Physics},
  volume={18},
  number={3},
  pages={033024},
  year={2016},
  publisher={IOP Publishing}
}

@book{QSTparis2004quantum,
  title={Quantum state estimation},
  author={Paris, Matteo and Rehacek, Jaroslav},
  volume={649},
  year={2004},
  publisher={Springer Science \& Business Media}
}

@book{QSTenglert2025lecture,
  title={Lecture on Quantum State Estimation},
  author={Englert, Berthold-Georg},
  year={2025},
  publisher={World Scientific}
}

@article{QSTcramer2010efficient,
  title={Efficient quantum state tomography},
  author={Cramer, Marcus and Plenio, Martin B and Flammia, Steven T and Somma, Rolando and Gross, David and Bartlett, Stephen D and Landon-Cardinal, Olivier and Poulin, David and Liu, Yi-Kai},
  journal={Nature communications},
  volume={1},
  number={1},
  pages={149},
  year={2010},
  publisher={Nature Publishing Group UK London}
}

@article{hradil1997quantum,
  title={Quantum-state estimation},
  author={Hradil, Zdenek},
  journal={Physical Review A},
  volume={55},
  number={3},
  pages={R1561},
  year={1997},
  publisher={APS}
}

@article{blume2010optimal,
  title={Optimal, reliable estimation of quantum states},
  author={Blume-Kohout, Robin},
  journal={New Journal of Physics},
  volume={12},
  number={4},
  pages={043034},
  year={2010}
}

@article{vrehavcek2007diluted,
  title={Diluted maximum-likelihood algorithm for quantum tomography},
  author={{\v{R}}eh{\'a}{\v{c}}ek, Jaroslav and Hradil, Zden{\v{e}}k and Knill, Emanuel and Lvovsky, Alexander I},
  journal={Physical Review A—Atomic, Molecular, and Optical Physics},
  volume={75},
  number={4},
  pages={042108},
  year={2007},
  publisher={APS}
}

@article{oberender2025spurious,
  title={On spurious fixed points in iterative maximum likelihood reconstruction for quantum tomography},
  author={Oberender, Florian},
  journal={arXiv preprint arXiv:2508.14549},
  year={2025}
}

@article{fuchs2003quantum,
  title={Quantum mechanics as quantum information, mostly},
  author={Fuchs, Christopher A},
  journal={Journal of Modern Optics},
  volume={50},
  number={6-7},
  pages={987--1023},
  year={2003},
  publisher={Taylor \& Francis}
}

@article{liu2025retrodictive,
  title={Retrodictive approach to quantum state smoothing},
  author={Liu, Mingxuan and Scarani, Valerio and Auff{\`e}ves, Alexia and Laverick, Kiarn T},
  journal={Physical Review A},
  volume={112},
  number={3},
  pages={L030203},
  year={2025},
  publisher={APS}
}

@article{liu2025unifying,
  title={Unifying quantum smoothing theories with extended retrodiction},
  author={Liu, Mingxuan and Bai, Ge and Scarani, Valerio},
  journal={arXiv preprint arXiv:2510.08447},
  year={2025}
}

\appendix
\onecolumngrid
\section{Parameter-Space Lifting, Bit Example}\label{app:coin}
Examples of event-wise classical Jeffrey's updating for characterizing probability functions exist in many standard texts on statistics (consider \cite{wasserman2004all-of-statistics}). We include one here briefly, just for illustration. Consider a bit with bias $b\in[0,1]$, so that
\begin{equation}
    p(0|b)=1-b,\qquad p(1|b)=b.
\end{equation}
The parameter $b$ therefore defines a classical channel
\begin{equation}
    \tilde{\mathcal E}(y|b)=p(y|b),
    \qquad y\in\{0,1\}.
\end{equation}
Here $y$ denotes the outcome of a single click. Rather than assigning a
belief directly to the probability vector $(1-b,b)$, we assign a prior
density $\pi_0(b)$ over the possible biases.

After $N$ clicks, let $n$ denote the number of clicks with outcome $0$, so
that $N-n$ have outcome $1$. The likelihood of the observed data is
\begin{equation}
    \prod_{\ell=1}^N
    \tilde{\mathcal E}([q_\ell]|b)
    =(1-b)^n b^{N-n}.
\end{equation}
Thus the event-wise Jeffrey's updates give
\begin{equation}
    \pi_N(b)
    =
    \frac{
        \pi_0(b)(1-b)^n b^{N-n}
    }{
        \int_0^1 db'\,
        \pi_0(b')(1-b')^n {b'}^{N-n}
    }.
\end{equation}
If the true bias is $b^\ast$, then $\frac{n}{N}\longrightarrow 1-b^\ast$ and $\frac{N-n}{N}\longrightarrow b^\ast$.
Consequently, provided $\pi_0(b^\ast)>0$, the posterior concentrates on the
true parameter,
\begin{equation}
    \lim_{N\to\infty}\pi_N(b)=\delta(b-b^\ast).
\end{equation}

\onecolumngrid

\section{Proof for Theorem \ref{thm:c-s-conv}}
\label{app:c-s-conv}
The principle of Theorem \ref{thm:c-s-conv} is largely a known consistency between frequentism and Bayesianism when the same evidence is inferred upon over an arbitrarily large number of rounds \cite{wasserman2004all-of-statistics}. We express it here explicitly in the category of finite stochastic maps. 

All that is needed is to frame the iteration within Csiszár and Tusnády's information-geometric framework \cite{csiszar1984information}. First, let us state their theorem.
    
\begin{theorem}\label{thm:c-t-conv}
    \cite{csiszar1984information} $\mathcal{P}$ and $\mathcal{Q}$ are two convex sets of finite measures. Let $P_0 \in \mathcal{P}$ be some arbitrary starting element, and let for each $n \geq 0$, $Q_n = \argmin_{Q\in\mathcal{Q}}D(P_n \parallel  Q)$ and $P_{n+1} = \argmin_{P\in\mathcal{P}}D(P \parallel  Q_n)$. Then, $D(P_n\parallel Q_n)$ converges to the infimum of $D(P\parallel Q)$ on $\mathcal{P}_0 \times \mathcal{Q}$ where $\mathcal{P}_0$ is the set of all $P \in \mathcal{P}$ such that $D(P\parallel Q_n)<\infty$ for some $n$.
\end{theorem}
    
One of the stated motivation for their theorem is for the Expectation-Maximisation (EM) algorithm for MLE problems, which iterative Jeffrey's update is an instance of. Thus, here we will state the explicitly the connection between Jeffrey's update and Theorem~\ref{thm:c-t-conv}. First, let us first begin by defining two convex sets,
\begin{equation}
    \mathcal{Q}=\{\chn(y|x)p(x) \;|\; \forall p(x)\in \Delta_d\}    
\end{equation}
and
\begin{equation}
    \mathcal{P}=\{J(x,y)\;|\;\sum_x J(x,y)=q(y)\}.
\end{equation}
For simplicity $P(x,y)$ will be understood as a member of $\mathcal{P}$, and $P(x)$ and $P(y)$ will be understood as its marginals, and similarly for $Q$ and $\mathcal{Q}$. Now, note that $P_{n+1} = \argmin_{P\in\mathcal{P}}D(P \parallel  Q_n)$ is simply finding the closest joint distribution in $\mathcal{P}$ given a point $Q_n \in \mathcal{Q}$ which can be expressed as $Q_n(x,y)=\chn(y|x)\pest[n](x)$. To project this to $\mathcal{P}$, we need to solve or
\begin{align}
    D(P\parallel Q_n) &= \sum_{x,y}P(x,y) \log \frac{P(x,y)}{\chn(y|x)\pest[n](x)} \nonumber \\
    &= \sum_{x,y}P(y)P(x|y) \log \frac{P(y)P(x|y)}{Q_n(x|y)Q_n(y)} \nonumber \\
    &= \sum_y P(y) \log \frac{P(y)}{Q_n(y)} + \sum_{y} P(y)\sum_xP(x|y)\log\frac{P(x|y)}{Q_n(x|y)}.
\end{align}
The first term is fixed because $P(y)=q(y)$ by the constraint on $\mathcal{P}$. What remains is the term $P(x|y)$ in $\sum_xP(x|y)\log\frac{P(x|y)}{Q_n(x|y)}$. This is simply a relative entropy of the marginals for a fixed $y$, and it is minimised by $P(x|y) = Q_n(x|y)$. Thus,
\begin{equation}
    P_{n+1}(x,y)=Q_n(x|y)q(y)=\frac{\chn(y|x)\pest[n](x)}{[\chn\pest[n]](y)}q(y).
\end{equation}
Note that when seen as an EM algorithm, this first projection is the E-step. On the other hand, when we project it back to $\mathcal{Q}$, we need to find $p$ such that $D(P_{n+1}\parallel \chn p)$ is minimised.
\begin{align}
    D(P_{n+1}\parallel \chn p)&=\sum_{x,y} P_{n+1}(x,y) \log \frac{P_{n+1}(x,y)}{\chn(y|x)p(x)} \nonumber \\
    &= \sum_{x,y} P_{n+1}(x,y) \log \frac{P_{n+1}(x,y)}{\chn(y|x)} - \sum_{x} P_{n+1}(x)\log p(x). 
\end{align}
The first term is constant and thus we can ignore. By Gibbs' inequality, the term $- \sum_{x} P_{n+1}(x)\log p(x)$ is minimal when $p(x) = P_{n+1}(x)=\sum_yP_{n+1}(x,y)$. Thus,
\begin{equation}
    Q_{n+1}(x,y)=\chn(y|x)\sum_{y'}P_{n+1}(x,y').
\end{equation}
In particular, when we consider the $x$ marginal of $Q_{n+1}$, we obtain back Jeffrey's update:
\begin{equation}
    \sum_yQ_{n+1}(x,y)=\sum_{y'}P_{n+1}(x,y')=\sum_y\frac{\chn(y|x)\pest[n](x)}{[\chn\pest[n]](y)}q(y).
\end{equation}
This second projection is the M-step. Therefore, iterating over Jeffrey's update is simply a series of projections between the two convex sets $\mathcal{P}$ and $\mathcal{Q}$. Now, observe that, after the projection to $\mathcal{P}$,
\begin{align}
    D(P_{n+1}\parallel Q_n) &=\sum_y P_{n+1}(y) \log \frac{P_{n+1}(y)}{Q_n(y)} + \sum_{y} P_{n+1}(y)\sum_xP_{n+1}(x|y)\log\frac{P_{n+1}(x|y)}{Q_n(x|y)} \nonumber \\
    &= \sum_y q(y) \log \frac{q(y)}{[\chn\pest[n]](y)} \nonumber \\
    &= \sum_y q(y) \log q(y) - \mathcal{L}(q|\pest[n]).
\end{align}
where here, the second equality is obtained from the fact that $P_{n+1}(x|y)=Q_n(x|y)$ and $\mathcal{L}(q|\pest[n]) = \sum_y q(y) \log [\chn\pest[n]](y) $ is the log-likelihood of observing $q(y)$ given an estimator $\pest[n]$. Note that the first term in the last line is a constant term. Thus, from Theorem~\ref{thm:c-t-conv}, the iteration of Jeffrey's updates will bring $D(P_{n+1}\parallel Q_n)$ to the infimum of $D(P\parallel Q)$ over $\mathcal{P}_0\times\mathcal{Q}$. Equivalently, this infimum divergence is related to the supremum of the log-likelihood $\mathcal{L}(q|\pest[n])$ by
\begin{equation}
    \inf_{(P,Q) \in \mathcal{P}_0\times\mathcal{Q}} D(P\parallel Q) = \sum_y q(y)\log q(y) - \sup_{\pest \in \Delta_d} \mathcal{L}(q|\pest).
\end{equation}
In particular, by the fact that in the asymptotic regime the evidence $q=\chn p$ is the image of the true state $p$, and that because of the initial estimator $\pest[0]$ is full-support we have $\text{supp}(p) \subseteq \text{supp}(\pest[0])$, the supremum of the log-likelihood itself is $\sum_y q(y)\log q(y)$. Thus, for $n \rightarrow \infty$, $D(P_n\parallel Q_n) \rightarrow 0$ and $\chn \pest[n] \rightarrow q$. Finally, if $\chn$ is injective, $\pest[n] \rightarrow p$.

\section{Iterative Petz vs Iterative Mirrored Petz}\label{app:failure}

The crux of the difference in the two updates can be seen when we write both updates in the eigenbasis of $\rhoest$. At some current estimator $\rhoest$, via the Petz map we have
\begin{equation}
    \left(\sqrt{\rhoest}R\sqrt{\rhoest}\right)_{j,k}=\sqrt{\lambda_j \lambda_k} R_{jk},
\end{equation}
whereas via the mirrored Petz map we have
\begin{equation}
    \left(\sqrt{R}\rhoest\sqrt{R}\right)_{j,k}= \sum_l \left(\sqrt{R}\right)_{j,l} \lambda_l \left(\sqrt{R}\right)_{l,k},
\end{equation}
where in both cases $\lambda_j$ is the eigenvalue of $\rhoest$ associated with $\ket{j}$. Note that at the true state, $R_{ij}=\delta_{i,j}$

Now, as noted in Section~\ref{sec:numerics}, the iterative Petz often fails when the true state has high purity. This means, at the region close to the true state, some of the eigenvalues of $\rhoest$ are small. Let $\lambda_j \sim \epsilon$ be one of these small eigenvalues. The $j,k$-th entry of the Petz map update will then only update by $O(\sqrt{\epsilon})$. On the other hand, for the mirrored Petz map, the fact that the summation over all eigenvalues is always in the expression means that the $j,k$-th entry will update by $O(|\sqrt{R}|_{jk})$. I.e., the coherence terms of the Petz map get eliminated at region with high purity. This means that, at this regime, the Petz map is unable to rotate the basis of the estimator and is locked in the same basis.

At this point, note that the eigenvalues update by $\lambda_j' = m_j \lambda_j$ for some multiplier $m_j$. Any fixed points need $m_j=1$ or $\lambda_j=0$ for all eigenvalues. At the true state, the multiplier is given by $R_{jj}$, which evaluates to 1 at the true state. However, if the eigenbasis could not be aligned, this is not possible, so the only way for a fixed point to exist is by having $\lambda_j=0$ for the misaligned space, which leads to a rank-deficient state.

Figure~\ref{fig:angle} illustrates this failure. Starting from a point at a fixed trace distance from the true state, the failure rate of the iterative Petz map increases as the principal angle between the \textit{weak} eigenspaces, i.e., corresponding to small eigenvalues, of the initialising estimator and the true state gets higher. On the other hand, the mirrored Petz map has no such problem, converging in every instance.

\begin{figure}
    \centering
    \includegraphics[width=0.5\textwidth]{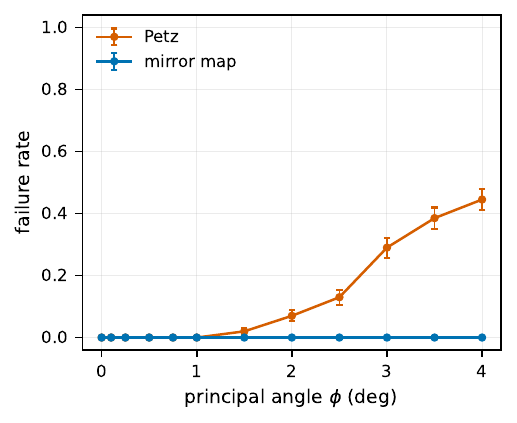}
    \caption{The relation between the failure rate and the principal angle between the weak eigenspaces of the estimator and of the true state at a fixed trace distance of $0.05$ and $d=4$. The condition of failure is identical to Fig.~\ref{fig:capture_rate}. The same 200 true states and POVMs are used at every data point, with the starting states perturbed to the same trace distance but at different angles. $\phi=0$ corresponds to a rotation acting only within the strong and weak blocks separately.}
    \label{fig:angle}
\end{figure}

\end{document}